\PassOptionsToPackage{dvipsnames}{xcolor}
\documentclass[runningheads, envcountsame]{llncs}
\usepackage[T1]{fontenc}
\usepackage[utf8]{inputenc}
\usepackage{bussproofs}
\usepackage{tikz}
\usepackage{tikz-cd}
\usepackage{mathrsfs}
\usepackage{amsmath}
\usepackage{amssymb}
\usepackage{stmaryrd}
\usepackage{mathtools}
\usepackage{yfonts}

\usepackage{caption}
\usepackage[b]{esvect}

\DeclareMathAlphabet{\mathpzc}{OT1}{pzc}{m}{it}

\spnewtheorem{assumption}[theorem]{Assumption}{\bfseries}{\itshape}
\spnewtheorem{notation}[theorem]{Notation}{\bfseries}{\itshape}

\newenvironment{bprooftree}
  {\leavevmode\hbox\bgroup}
  {\DisplayProof\egroup}

\newcommand{\secref}[1]{\S\ref{#1}}%

\newcommand\lang{\textbf{Aff}}

\newcommand\ket[1]{{|{#1}\rangle}}

\newcommand{\id}{\text{id}}
\newcommand{\Id}{\text{Id}}

\newcommand{\op}{\ensuremath{\mathrm{op}}}

\newcommand{\Ob}{\mathrm{Ob}}

\newcommand{\WNMIU}{\ensuremath{\mathbf{W}^*_{\mathrm{NMIU}}}}

\newcommand{\DD}{\ensuremath{\mathbf{D}}}

\newcommand{\AAA}{\ensuremath{\mathbf{A}}}
\newcommand{\BB}{\ensuremath{\mathbf{B}}}
\newcommand{\CC}{\ensuremath{\mathbf{C}}}
\newcommand{\VV}{\ensuremath{\mathbf{V}}}
\newcommand{\CCc}{\ensuremath{\mathbf{C}_a}}
\newcommand{\JJ}{\ensuremath{\mathbf{J}}}
\newcommand{\Wstar}{\ensuremath{\mathbf{W}^*_{\mathrm{NCPSU}}}}

\newcommand{\cpo}{\ensuremath{\mathbf{DCPO}}}
\newcommand{\dcpobs}{\ensuremath{\mathbf{DCPO}_{\perp !}}}
\newcommand{\cpobs}{\ensuremath{\mathbf{DCPO}_{\perp !}}}

\newcommand{\bit}{\textbf{bit}}

\newcommand{\newunit}{\textbf{new unit}}

\newcommand{\discard}{\textbf{discard}}
\newcommand{\sskip}{\textbf{skip}}

\newcommand{\while}[2]{\ensuremath{\textbf{while}\ $#1$\ \textbf{do}\ $#2$}}
\newcommand{\lleft}{\textbf{left}}
\newcommand{\rright}{\textbf{right}}
\newcommand{\ccase}{\textbf{case}}
\newcommand{\fold}{\textbf{fold}}

\newcommand{\sfold}{\ensuremath{\mathrm{fold}}}
\newcommand{\sunfold}{\ensuremath{\mathrm{unfold}}}

\newcommand{\efold}{\textswab{fold}}
\newcommand{\eunfold}{\textswab{unfold}}

\newcommand{\ffold}{\textbf{fold}}

\newcommand{\lrb}[1]{{\llbracket #1 \rrbracket}}

\newcommand{\elrb}[1]{{  \talloblong #1 \talloblong }}

\newcommand{\naturalto}{\ensuremath{\Rightarrow}}

\usetikzlibrary{decorations.pathmorphing}
\usetikzlibrary{decorations.markings}
\usetikzlibrary{decorations.pathreplacing}
\usetikzlibrary{arrows}
\usetikzlibrary{shapes}

\pgfdeclarelayer{edgelayer}
\pgfdeclarelayer{nodelayer}
\pgfsetlayers{edgelayer,nodelayer,main}

\tikzstyle{braceedge}=[decorate,decoration={brace,amplitude=10pt}]
\tikzstyle{square box}=[rectangle,fill=white,draw=black,minimum height=6mm,minimum width=6mm,yshift=0.7mm]
\tikzstyle{wire label}=[font=\footnotesize, auto,swap]

\tikzstyle{none}=[inner sep=0pt]
\tikzstyle{gn}=[circle,fill=Lime,draw=Black,line width=0.8 pt]
\tikzstyle{rn}=[circle,fill=Red,draw=Black, line width=0.8 pt]
\tikzstyle{H}=[rectangle,fill=Yellow,draw=Black]
\tikzstyle{line}=[scalar,fill=White,draw=Black]
\tikzstyle{io}=[rectangle,fill=White,draw=Black]
\tikzstyle{block}=[rectangle,fill=Orange,draw=Black]
\tikzstyle{graph}=[circle,fill=White,draw=Black]
\tikzstyle{empty}=[rectangle,fill=none,draw=none]
\tikzstyle{scaled}=[rectangle,fill=none,draw=none, font=\small]
\tikzstyle{box}=[rectangle,fill=White,draw=Black]
\tikzstyle{dot}=[circle,fill=Black,draw=Black,inner sep=0pt,minimum size=1pt]
\tikzstyle{small dot}=[circle,fill=Black,draw=Black,inner sep=0pt,minimum size=1pt]
\tikzstyle{Dot}=[circle,fill=Black,draw=Black,inner sep=0pt,minimum size=3pt]
\tikzstyle{diam}=[rectangle,fill=Black,draw,yscale=1.2,rotate=45]
\tikzstyle{gangle}=[rectangle,fill=Lime,draw=Black]
\tikzstyle{rangle}=[rectangle,fill=Red,draw=Black]
\tikzstyle{circ}=[circle,fill=none,draw=Black,scale=1.3]
\tikzstyle{ellip}=[ellipse,fill=none,draw=Black,scale=1.3,minimum width =1.3cm]
\tikzstyle{ellip2}=[ellipse,fill=White,draw=Black,scale=1.3,minimum width =3cm]
\tikzstyle{bbox}=[rectangle,fill=Blue,draw=Blue,scale=0.6]
\tikzstyle{gg}=[shape=rectangle,fill=White,draw=Black,dashed]

\tikzstyle{white circle}=[circle,fill=none,draw=Black,scale=1]
\tikzstyle{black circle}=[circle,fill=Black,draw=Black,scale=1]
\tikzstyle{grey circle}=[circle,fill=Gray,draw=Black,scale=1]
\tikzstyle{white rectangle}=[rectangle,fill=none,draw=Black,scale=1]

\tikzstyle{nodev}=[circle,fill=none,draw=Black,scale=1]
\tikzstyle{greynode}=[circle,fill=Grey,draw=Black,scale=1]
\tikzstyle{blacknode}=[circle,fill=Black,draw=Black,scale=1]
\tikzstyle{wirev}=[circle,fill=Black,draw=Black,inner sep=0pt,minimum size=3pt]
\tikzstyle{wirevred}=[circle,fill=Red,draw=Black,inner sep=0pt,minimum size=3pt]

\tikzstyle{simple}=[-,draw=Black]
\tikzstyle{to}=[->,draw=Black]
\tikzstyle{naturalto}=[-{Implies},double distance=1.5pt]
\tikzstyle{bdirected}=[<->,draw=Black]
\tikzstyle{bothdirs}=[bdirected,draw=Black]
\tikzstyle{bothdirsred}=[bdirected,draw=Red]
\tikzstyle{blue}=[-,draw=Blue]
\tikzstyle{redd}=[directed,draw=Red]
\tikzstyle{redu}=[-,draw=Red]
\tikzstyle{blued}=[directed,draw=Blue]
\tikzstyle{dash}=[dashed,draw=Black]
\tikzstyle{ddash}=[->,dashed,draw=Black]
\tikzstyle{dashedd}=[->,dashed]
\tikzstyle{dashedred}=[dashed,draw=Red]

\tikzstyle{equal-arrow}=[double equal sign distance]

\tikzstyle{dotpic}=[scale=0.5]

\tikzstyle{every picture}=[baseline=-0.25em]

\newcommand{\stikz}[2][1]{\scalebox{#1}{
\InputIfFileExists{#2}{}{\input{./tikz/#2}}
}}
\newcommand{\cstikz}[2][1]{\begin{center}\stikz[#1]{#2}\end{center}}

\begin{document}

\title{Semantics for first-order affine inductive data types via slice categories}

\author{Vladimir Zamdzhiev}
\authorrunning{V. Zamdzhiev}

\institute{Universit\'e de Lorraine, CNRS, Inria, LORIA, France}

\maketitle

\begin{abstract}
Affine type systems are substructural type systems where copying of information
is restricted, but discarding of information is permissible at all types. Such
type systems are well-suited for describing quantum programming languages,
because copying of quantum information violates the laws of quantum mechanics.
In this paper, we consider a first-order affine type system with inductive
data types and present a novel categorical semantics for it. The most
challenging aspect of this interpretation comes from the requirement to
construct appropriate discarding maps for our data types which might be defined
by mutual/nested recursion. We show how to achieve this for all types by taking
models of a first-order linear type system whose atomic types are discardable
and then presenting an additional affine interpretation of types within the
slice category of the model with the tensor unit. We present some concrete
categorical models for the language ranging from classical to quantum. Finally,
we discuss potential ways of dualising and extending our methods and using them
for interpreting coalgebraic and lazy data types.
\end{abstract}

\section{Introduction}

Linear Logic~\cite{linear-logic} is a substructural logic where the rules for
weakening and contraction are restricted. Linear logic has been very
influential in computer science and has lead to the development of linear type
systems where discarding and copying of variables is restricted. Linear logic
has also inspired the development of \emph{affine} type systems, which are
substructural type systems where only the rule for contraction (copying of
variables) is restricted, but weakening (discarding of variables) is completely
unrestricted. Affine type systems are a natural choice for quantum programming
languages~\cite{qpl,qpl-fossacs,qlc-affine,quantum-games}, because they can be
used to enforce compliance with the laws of quantum mechanics, where copying of
quantum information is impossible~\cite{no-cloning}.

In this paper we consider a first-order affine type system with inductive
data types, called \lang, and we present a categorical semantics for it. The
main focus of the present paper is on the construction of the required
discarding maps that are necessary for the interpretation of the type system.
Our semantics is novel in that we assume very little structure on the model
side: we do not assume the existence of any (sub)category where the tensor unit
$I$ is a terminal object. Instead, we merely assume that the interpretation of
every \emph{atomic} type is equipped with some discarding map (which is clearly
necessary) and we then show how to construct \emph{all other} discarding maps
by providing an \emph{affine interpretation of types} within the slice category
of the model with the tensor unit. Thus, by taking a categorical model of a
first-order linear type system, we construct all the discarding maps we need by
performing a careful \emph{semantic} analysis, instead of assuming additional
structure within the categorical model.

\paragraph{Outline.}
We begin by recalling some background about parameterised initial algebras in
Section~\ref{sec:back}. Next, we describe the syntax of \lang, which is a
fragment of the quantum programming language QPL~\cite{qpl,qpl-fossacs}, in
Section~\ref{sec:syntax}. In Section~\ref{sec:operational}, we present the
operational semantics of \lang.  One of our main contributions is in
Section~\ref{sec:discarding}, where we show how parameterised initial algebras
for suitable functors may be reflected into slice categories.  Our
contributions continue in Section~\ref{sec:model}, where we describe a
categorical model for our language, and with Section~\ref{sec:semantics},
where we present a novel categorical semantics for the affine structure of
types by providing a non-standard type interpretation within a slice
category.  In Section~\ref{sec:future} we discuss future work and possible
extensions and in Section~\ref{sec:conclude} we discuss related work and
present some concluding remarks.

\section{Parameterised Initial Algebras}
\label{sec:back}

Simple inductive data types, like lists and natural numbers, may be interpreted
by initial algebras. However, the interpretation of inductive data types defined
by mutual/nested induction requires a more general notion called
\emph{parameterised initial algebra}, which we shall now recall.

\begin{definition}[cf. {\cite[\S 6.1]{fiore-thesis}}]
  Let $\AAA$ and $\BB$ be categories and $T : \AAA \times \BB \to \BB$ a functor. A \emph{parameterised initial algebra}
  for $T$ is a pair $(T^\dagger, \tau),$ such that:
  \begin{itemize}
    \item $T^\dagger : \AAA \to \BB$ is a functor;
    \item $\tau : T \circ \langle \Id, T^\dagger \rangle \naturalto T^\dagger : \AAA \to \BB$ is a natural isomorphism;
    \item For every $A \in \Ob(\AAA)$, the pair $(T^\dagger A, \tau_A)$ is an initial $T(A, -)$-algebra.
  \end{itemize}
\end{definition}
Note that by trivialising $\AAA$, we get the well-known notion of initial algebra.
Next, we recall a theorem which provides sufficient conditions for the existence of parameterised initial algebras.
\begin{theorem}[{\cite[Theorem 4.12]{lnl-fpc-lmcs}}]
\label{thm:par-exists}
Let $\BB$ be a category with an initial object and all $\omega$-colimits. Let $T: \AAA \times \BB \to \BB$ be an $\omega$-cocontinuous functor. Then
$T$ has a parameterised initial algebra $(T^\dagger, \tau)$ and the functor $T^\dagger$ is also $\omega$-cocontinuous.
\end{theorem}
In particular, the above theorem shows that $\omega$-cocontinuous functors are closed under formation of parameterised initial algebras.

\section{Syntax of \lang}\label{sec:syntax}

In this section we describe the syntax of \lang, which is the language on which we will base the development of our ideas. \lang{} is a fragment of the quantum programming language QPL~\cite{qpl-fossacs} which is obtained from QPL by removing procedures, quantum resources and copying of classical information.
The reason for considering this fragment is just simplicity and brevity of the presentation.

\begin{remark}
In fact, the methods we describe can handle the addition of
procedures and the copying of non-linear information with no further effort.
The addition of quantum
resources can also be handled by our methods, but this requires identifying a
suitable category of quantum computation with $\omega$-colimits.
\end{remark}
The syntax of \lang{} is summarised in Figure~\ref{fig:syntax}.
A type context $\Theta$, is \emph{well-formed}, denoted $\vdash \Theta$, if $\Theta$ is simply a list of distinct type variables.
Well-formed types, denoted $\Theta \vdash A,$ are specified by the following rules:
\[
{\small{
    \begin{bprooftree}
    \AxiomC{{$\vdash \Theta$}}
    \UnaryInfC{$\Theta \vdash \Theta_i$}
    \end{bprooftree}
    \begin{bprooftree}
    \AxiomC{\phantom{$\Theta \vdash A$}}
    \UnaryInfC{$\Theta \vdash I$}
    \end{bprooftree}
    \begin{bprooftree}
    \AxiomC{$\mathbf A$ atomic}
    \UnaryInfC{$\Theta \vdash \mathbf A$}
    \end{bprooftree}
    \begin{bprooftree}
    \AxiomC{$\Theta \vdash A$}
    \AxiomC{$\Theta \vdash B$}
    \RightLabel{$\star \in \{+, \otimes\}$}
    \BinaryInfC{$\Theta \vdash A \star B$}
    \end{bprooftree}
    \begin{bprooftree}
    \AxiomC{$\Theta, X \vdash A$}
    \UnaryInfC{$\Theta \vdash \mu X.A$}
    \end{bprooftree},
}}
\]
where we assume that there is some set of atomic types $\mathcal A$, which we will leave
unspecified in this paper (for generality). For example, in quantum programming, it
suffices to assume $\mathcal A = \{ \textbf{qubit} \}$. This is the case for QPL.

A type $A$ is \emph{closed} whenever $\cdot \vdash A$. Note that nested type induction (also known as mutual induction) is allowed, i.e., it is possible to form inductive data types which have more than one free variable in their type contexts.
Henceforth, we implicitly assume that all types we are dealing with are well-formed.

%
%
\begin{figure}[t]
{\small{
\begin{tabular}{l  l  l  l}
  Type Variables           & $X,Y,Z$ & & \\
  Term Variables           & $x, y, b, u$ & & \\
	Atomic Types             & $\mathbf A \in \mathcal A$ &  &   \\
	Types                    & $A, B, C$ & ::= &  $X$ | $I$ | \textbf{A} | $A+B$ | $A\otimes B$ | $\mu X. A$\\
  Terms                    & $M, N$ & ::= & \newunit\ $u$ | \discard\ $x$ | $  M;N$ | \sskip\  | \\
    & &                                   &  \while{$b$}{$M$} | $x = $ \textbf{left}$_{A,B} M$ | $x$ = \textbf{right}$_{A,B} M$ | \\ 
    & &                                   & \textbf{case} $y$ \textbf{of} $\{$\textbf{left} $x_1\to M\ |$ \textbf{right} $x_2 \to N\}$ | \\ 
    & &                                   & $x = (x_1, x_2) $ | $(x_1, x_2) = x$ | $y$ = \textbf{fold} $x$ | $y$ = \textbf{unfold} $x$ \\
  Type contexts            & $\Theta$ &  ::= & $X_1, X_2, \ldots, X_n $\\
  Variable contexts        & $\Gamma, \Sigma$ &  ::= & $x_1: A_1, \ldots, x_n : A_n$\\
  Type Judgements          & \multicolumn{2}{l}{$\Theta \vdash A$} &\\
  Term Judgements          & \multicolumn{2}{l}{$\vdash \langle \Gamma \rangle\ M\ \langle \Sigma \rangle$} &
\end{tabular}
}}
\caption{Syntax of \lang.}
\label{fig:syntax}
\end{figure}
%
%
\begin{figure}[t]
{
\small{
  \[
    \begin{bprooftree}
    \AxiomC{}
    \UnaryInfC{$ \vdash \langle \Gamma \rangle\ \newunit\ u\ \langle \Gamma, u:I \rangle$}
    \end{bprooftree}
    \begin{bprooftree}
    \AxiomC{}
    \UnaryInfC{$ \vdash \langle \Gamma, x:A \rangle\ \textbf{discard}\ x\ \langle \Gamma \rangle$}
    \end{bprooftree}
  \]

  \[
    \begin{bprooftree}
    \AxiomC{\phantom{P}}
    \UnaryInfC{$ \vdash \langle \Gamma \rangle\ \sskip\ \langle \Gamma \rangle$}
    \end{bprooftree}
  \]
  
  \[
    \begin{bprooftree}
    \AxiomC{$ \vdash \langle \Gamma \rangle\ M\ \langle \Gamma' \rangle$}
    \AxiomC{$ \vdash \langle \Gamma' \rangle\ N\ \langle \Sigma \rangle$}
    \BinaryInfC{$ \vdash \langle \Gamma \rangle\ M;N\ \langle \Sigma \rangle$}
    \end{bprooftree}
    \begin{bprooftree}
    \AxiomC{$ \vdash \langle \Gamma, b: \textbf{bit} \rangle\ M\ \langle  \Gamma, b: \textbf{bit} \rangle$}
    \UnaryInfC{$ \vdash \langle \Gamma, b: \textbf{bit} \rangle\ 
      \while{$b$}{$M$}\ \langle \Gamma, b: \textbf{bit} \rangle$}
    \end{bprooftree}
   \]

   \[
    \begin{bprooftree}
    \AxiomC{}
    \UnaryInfC{$ \vdash \langle \Gamma, x:A \rangle\ y = \lleft_{A,B}\ x\ \langle \Gamma, y: A+B \rangle$}
    \end{bprooftree}
    \begin{bprooftree}
    \AxiomC{}
    \UnaryInfC{$ \vdash \langle \Gamma, x:B \rangle\ y = \rright_{A,B}\ x\ \langle \Gamma, y: A+B \rangle$}
    \end{bprooftree}
  \]

  \[
    \begin{bprooftree}
    \AxiomC{$ \vdash \langle \Gamma, x_1: A \rangle\ M_1\ \langle \Sigma \rangle$}
    \AxiomC{$ \vdash \langle \Gamma, x_2: B \rangle\ M_2\ \langle \Sigma \rangle$}
    \BinaryInfC{$ \vdash \langle \Gamma, y: A+B \rangle\ 
      \ccase\ y\ \textbf{of}\ \{\lleft_{A,B}\ x_1 \to M_1\ |\ \rright_{A,B}\ x_2 \to M_2\ \}\ \langle \Sigma \rangle$}
    \end{bprooftree}
    \qquad\ \ 
  \]

  \[
    \begin{bprooftree}
    \AxiomC{}
    \UnaryInfC{$ \vdash \langle \Gamma, x_1: A, x_2 : B \rangle\ x = (x_1, x_2)\ \langle \Gamma, x: A \otimes B \rangle$}
    \end{bprooftree}
  \]

  \[
    \begin{bprooftree}
    \AxiomC{}
    \UnaryInfC{$ \vdash \langle \Gamma, x: A \otimes B \rangle\ (x_1, x_2) = x \ \langle \Gamma, x_1: A, x_2 : B \rangle$}
    \end{bprooftree}
  \]

  \[
    \begin{bprooftree}
    \AxiomC{}
    \UnaryInfC{$ \vdash \langle \Gamma, x: A[\mu X. A / X] \rangle\ y = \textbf{fold}_{\mu X. A}\ x\ \langle \Gamma, y: \mu X.A \rangle$}
    \end{bprooftree}
  \]

  \[
    \begin{bprooftree}
    \AxiomC{}
    \UnaryInfC{$ \vdash \langle \Gamma, x: \mu X. A \rangle\ y = \textbf{unfold}\ x\ \langle \Gamma, y: A[\mu X. A / X] \rangle$}
    \end{bprooftree}
  \]
  }
}
\caption{Formation rules for \lang{} terms.}
\label{fig:term-formation}
\end{figure}

\begin{example}
\label{ex:syntax}
Natural numbers can be defined as $\textbf{Nat} \equiv \mu X. I + X.$
A list of a closed type $\cdot \vdash A$ is defined by $\textbf{List}(A) \equiv \mu Y. I + A \otimes Y.$
\end{example}
\emph{Term variables} are denoted by small Latin characters (e.g. $x,y,u,b$).
In particular, $u$ ranges over
variables of unit type $I$, $b$ over variables of type $\textbf{bit} \coloneqq
I+I$ and $x,y$ range over arbitrary variables. \emph{Variable contexts} are denoted by capital Greek letters, such as $\Gamma$ and
$\Sigma$. Variable contexts contain only variables of closed types and are written as $\Gamma = x_1: A_1, \ldots, x_n:A_n.$

A \emph{term judgement} $\vdash \langle \Gamma \rangle\ M\ \langle \Sigma \rangle$ 
indicates that term $M$ is well-formed assuming an
input variable context $\Gamma$ and an output variable context $\Sigma$. All types within it are necessarily closed. 
The formation rules are shown in Figure~\ref{fig:term-formation}.

\begin{remark}
Because we are not concerned with any domain-specific applications in this paper, we leave the atomic types uninhabited.
Of course, any domain-specific extension should add suitable introduction and elimination rules for each atomic type.
In the case of QPL, the term language has to be extended with three terms -- one each for preparing a qubit in state $\ket 0$, 
applying a unitary gate to a term and finally measuring a qubit. See~\cite{qpl-fossacs} for more information.
\end{remark}

\section{Operational Semantics of \lang}\label{sec:operational}

The purpose of this section is to present the operational semantics of \lang{}.
We begin by introducing \emph{program configurations} which completely
and formally describe the current state of program execution. A program configuration
is a pair $(M\ |\ V)$, where $M$ is the term which remains to be executed and $V$ is a \emph{value assignment}, which
is a function that assigns values to variables that have already been introduced.

\paragraph{Value Assignments.}
\emph{Values} are expressions defined by the following grammar:
\[v, w ::= *\ |\ \lleft_{A,B} v\ |\ \rright_{A,B} v\ |\ (v,w)\ |\ \ffold_{\mu X.A} v . \]
The expression $*$ represents the unique value of unit type $I$.
Other particular values of interest are the canonical values of type \textbf{bit}, called \emph{false} and \emph{true}, which are formally defined by $\texttt{ff}:=
\textbf{left}_{I,I} *$ and $\texttt{tt}:= \textbf{right}_{I,I} *.$
They play an important role in the operational semantics.

The well-formed values, denoted $\vdash v : A$, are specified by the following rules:
  {\small{
  \[
    \begin{bprooftree}
    \AxiomC{\phantom{$Q$}}
    \RightLabel{} \UnaryInfC{$\vdash *: I$}
    \end{bprooftree}
    \begin{bprooftree}
    \AxiomC{$\vdash v : A$}
    \UnaryInfC{$\vdash \lleft_{A,B} v : A+B$}
    \end{bprooftree}
    \begin{bprooftree}
    \AxiomC{$\vdash v : B$}
    \UnaryInfC{$\vdash \rright_{A,B} v : A+B$}
    \end{bprooftree}
  \]
  \[
    \begin{bprooftree}
    \AxiomC{$ \vdash v : A$}
    \AxiomC{$ \vdash w : B$}
    \BinaryInfC{$ \vdash (v, w) : A \otimes B$}
    \end{bprooftree}
    \begin{bprooftree}
    \AxiomC{$ \vdash v : A[\mu X. A / X]$}
    \UnaryInfC{$ \vdash \fold_{\mu X.A} v : \mu X. A$}
    \end{bprooftree}
  \]}}
A \emph{value assignment} is simply a function from term variables to values.
We write value assignments as $V =\{ x_1 = v_1, \ldots, x_n = v_n \},$
where each $x_i$ is a variable and each $v_i$ is a value. We say that
$V$ is \emph{well-formed} in variable context $\Gamma = \{ x_1 : A_1, \ldots x_n : A_n \}$, denoted $\Gamma \vdash V,$ if $V$ has the same variables as $\Gamma$
and $\vdash v_i : A_i,$ for each $i \in \{ 1, \ldots n \}$.

\paragraph{Program configurations.}
A \emph{program configuration} is a couple $(M\ |\ V),$ where $M$ is a term and where $V$ is a value assignment.
A \emph{well-formed} program configuration, denoted $\Gamma; \Sigma \vdash (M \ |\ V)$, is a program configuration $(M\ |\ V)$, such that there exist (necessarily unique) $\Gamma, \Sigma$ with:
(1) $\vdash \langle \Gamma \rangle\ M\ \langle \Sigma \rangle$ is a well-formed term; and
(2) $ \Gamma \vdash V$ is a well-formed value assignment.
\begin{figure}[t]
\small{
  \[
    \begin{bprooftree}
    \AxiomC{}
    \UnaryInfC{$
      (\newunit\ u\ |\ V)
      \leadsto
      (\sskip\ |\ V, u=*)
    $}
    \end{bprooftree}
  \]
  \vspace{0.5mm}
  \[
    \begin{bprooftree}
    \AxiomC{}
    \UnaryInfC{$
      (\textbf{discard}\ x\ |\ V, x=v)
      \leadsto
      (\sskip\ |\ V)
    $}
    \end{bprooftree}
  \]
  \vspace{0.5mm}
  \[
    \begin{bprooftree}
    \AxiomC{}
    \UnaryInfC{$(\sskip;P\ |\ V) \leadsto (P\ |\ V)$}
    \end{bprooftree}
    \qquad
  \]
  \vspace{0.5mm}
  \[
    \begin{bprooftree}
    \AxiomC{$(P\ |\ V) \leadsto (P'\ |\ V')$}
    \UnaryInfC{$(P;Q\ |\ V) \leadsto (P';Q\ |\ V')$}
    \end{bprooftree}
    \qquad
  \]
  \vspace{0.5mm}
  \[
    \begin{bprooftree}
    \AxiomC{}
    \UnaryInfC{$ 
      (\while{$b$}{$M$}\ |\ V, b = v)
      \leadsto
      (\textbf{if}\ b\ \textbf{then}\ \{M;\while{$b$}{$M$}\}
 \ |\ V, b = v )
    $}
    \end{bprooftree}
  \]
  \vspace{0.5mm}
  \[
    \begin{bprooftree}
    \AxiomC{}
    \UnaryInfC{$
      (y = \lleft\ x\ |\ V, x = v )
      \leadsto
      (\sskip\ |\ V, y =\lleft\ v)
    $}
    \end{bprooftree}
  \]
  \vspace{0.5mm}
  \[
    \begin{bprooftree}
    \AxiomC{}
    \UnaryInfC{$
      (y = \rright\ x\ |\ V, x = v )
      \leadsto
      (\sskip\ |\ V, y =\rright\ v )
    $}
    \end{bprooftree}
  \]
  \vspace{0.5mm}
  \[
    \begin{bprooftree}
    \AxiomC{}
    \UnaryInfC{$
      (\ccase\ y\ \textbf{of}\ \{\lleft\ x_1 \to M_1\ |\ \rright\ x_2 \to M_2\ \}\ |\ V, y=\lleft\ v)
      \leadsto
      (M_1\ |\ V, x_1=v)
    $}
    \end{bprooftree}
    \qquad\ \ 
  \]
  \vspace{0.5mm}
  \[
    \begin{bprooftree}
    \AxiomC{}
    \UnaryInfC{$
      (\ccase\ y\ \textbf{of}\ \{\lleft\ x_1 \to M_1\ |\ \rright\ x_2 \to M_2\ \}\ |\ V, y=\rright\ v)
      \leadsto
      (M_2\ |\ V, x_2=v)
    $}
    \end{bprooftree}
    \qquad\ \ 
  \]
  \vspace{0.5mm}
  \[
    \begin{bprooftree}
    \AxiomC{}
    \UnaryInfC{$
    (x = (x_1, x_2)\ |\ V, x_1 = v_1, x_2 = v_2 )
    \leadsto
    (\sskip\ |\ V, x=(v_1,v_2))
    $}
    \end{bprooftree}
  \]
  \vspace{0.5mm}
  \[
    \begin{bprooftree}
    \AxiomC{}
    \UnaryInfC{$
    ((x_1, x_2) = x\ |\ V, x=(v_1,v_2))
    \leadsto
    (\sskip\ |\ V, x_1 = v_1, x_2 = v_2)
    $}
    \end{bprooftree}
  \]
  \vspace{0.5mm}
  \[
    \begin{bprooftree}
    \AxiomC{}
    \UnaryInfC{$
      (y = \textbf{fold}\ x\ |\ V, x = v)
      \leadsto
      (\sskip\ |\ V, y =\textbf{fold}\ v)
    $}
    \end{bprooftree}
  \]
  \vspace{0.5mm}
  \[
    \begin{bprooftree}
    \AxiomC{}
    \UnaryInfC{$
      (y = \textbf{unfold}\ x\ |\ V, x = \textbf{fold}\ v)
      \leadsto
      (\sskip\ |\ V, y = v)
    $}
    \end{bprooftree}
  \]
}
\caption{Small Step Operational semantics of \lang{}.}
\label{fig:operational-small-step}
\end{figure}

The (small step) operational semantics is defined as a function $(- \leadsto - )$ on program configurations
$(M\ |\ V)$ by induction on the structure of $M$ in Figure~\ref{fig:operational-small-step}.
Note that, in the rule for while loops, the term $\textbf{if}\ b\ \textbf{then}\ \{M\} $ is just syntactic sugar for $\ccase\ b\ \textbf{of}\ \{\lleft\ u \to b = \lleft\ u\ |\ \rright\ u \to b = \rright\ u;M\ \}$. 

\begin{theorem}[Subject reduction \cite{qpl-fossacs}]
\label{thm:subject}
If $\Gamma; \Sigma \vdash (M\ |\ V)$ and $(M\ |\ V) \leadsto (M', V')$,
then $\Gamma'; \Sigma \vdash (M', V')$, for some (necessarily unique) context $\Gamma'$.
\end{theorem}

\begin{assumption}
Henceforth, all configurations are assumed to be well-formed.
\end{assumption}
We shall use calligraphic letters $(\mathcal C, \mathcal D, \ldots)$ to denote configurations.
A \emph{terminal} configuration is a configuration $\mathcal C$, such that $\mathcal C = (\sskip, V)$.
\begin{theorem}[Progress \cite{qpl-fossacs}]
\label{thm:progress}
If $\mathcal C$ is a configuration, then either $\mathcal C$ is terminal or there exists a configuration $\mathcal D$, such that 
$\mathcal C \leadsto \mathcal D.$
\end{theorem}
\begin{remark}
Any domain-specific extension should, of course, also adapt the operational
semantics as necessary. In the case of QPL, this requires introducing new
reduction rules for the additional terms and extending the notion of
configuration with an extra component that stores the quantum data.
\end{remark}

\section{Slice Categories for Affine Types}\label{sec:discarding}

Our type system is affine and in order to provide a denotational interpretation
we have to construct discarding maps in our model at every type. This is
achieved in the following way: (1) for every closed type $A$ we provide
a standard interpretation $\lrb A \in \text{Ob}(\CC)$ in our model $\CC$;
(2) in addition, we provide a \emph{affine} type interpretation $\elrb A \in \text{Ob}(\CC/I)$ 
within the \emph{slice category with the tensor unit}, that is, for every type we carefully pick out a specific
discarding map; (3) we prove $\elrb A = (\lrb A, \diamond_A : \lrb A \to I)$ and show that our choice of discarding map $\diamond_A$ can discard all values of our language, as required.

The purpose of this section is to show the slice category $\CC/I$ has
sufficient categorical structure for the affine interpretation of types. Our
analysis is quite general and works for many affine scenarios.
Under some basic assumptions on $\CC$ we show
that $\CC/I$ inherits from $\CC$: a symmetric monoidal structure
(Proposition~\ref{prop:affine-monoidal}), finite coproducts
(Proposition~\ref{prop:affine-coproducts}) and (parameterised) initial algebras for a sufficiently
large class of functors (Theorem~\ref{thm:par-initial-algebras-same}).
\begin{assumption}
\label{assume:slice}
Throughout the remainder of the section we assume we are given a category $\CC$ and we fix an object $I \in \Ob(\CC)$.
Let $\CCc \coloneqq \CC / I$ be the slice category of $\CC$ with the fixed object $I.$
\end{assumption}
Thus, the objects of $\CCc$ are pairs $(A, \diamond_A),$ where $A \in \text{Ob}(\CC)$
and $\diamond_A : A \to I$ is a morphism of $\CC$.
Then, a morphism $f : (A, \diamond_A) \to (B, \diamond_B)$ of $\CCc$
is a morphism $f : A \to B$ of $\CC$, such that $\diamond_B \circ f = \diamond_A.$ Composition and identities are the same as in $\CC$.
We refer to the maps $\diamond_A$ as the \emph{discarding} maps and to the morphisms of $\CCc$ as \emph{affine} maps.
\begin{notation}
There exists an obvious forgetful functor $U : \CCc \to \CC$ given by $U(A, \diamond_A) = A$ and $U(f) = f$.
\end{notation}
The following (well-known) proposition will be used to show the existence of certain initial algebras in $\CCc$. 
For completeness, we provide a proof.
\begin{proposition}\label{prop:reflect}
The functor $U: \CCc \to \CC$ reflects small colimits.
\end{proposition}
\begin{proof}
In Appendix~\ref{proof:reflect}.
\qed\end{proof}
Next, we show how a symmetric monoidal structure on $\CC$ induces one on $\CCc$.
\begin{proposition}\label{prop:affine-monoidal}
Assume that $\CC$ is equipped with a (symmetric) monoidal structure
$(\CC, \otimes, I, \alpha, \lambda, \rho, (\sigma))$. Then,
the tuple $(\CCc, \otimes_a, (I, \id_I), \alpha_a, \lambda_a, \rho_a, (\sigma_a))$ is a (symmetric) monoidal category, where
$\otimes_a : \CCc \times \CCc \to \CCc$ is defined by:
\begin{align*}
(A, \diamond_A) \otimes_a (B, \diamond_B) &\coloneqq (A \otimes B, \lambda_I \circ (\diamond_A \otimes \diamond_B) ) \\
f \otimes_a g &\coloneqq f \otimes g
\end{align*}
and where the natural isomorphisms $\alpha_a, \lambda_a, \rho_a, (\sigma_a)$ are componentwise equal to $\alpha, \lambda, \rho, (\sigma)$ respectively.
Moreover, this data makes $U: \CCc \to \CC$ a strict monoidal functor and we also have:
\[ \otimes \circ (U \times U ) = U \circ \otimes_a  : \CCc \times \CCc \to \CC . \]
\end{proposition}
\begin{proof}
Straightforward verification.
\qed\end{proof}
Next, we show how coproducts on $\CC$ induce coproducts on $\CCc$.
\begin{proposition}\label{prop:affine-coproducts}
Assume that $\CC$ has finite coproducts with initial object denoted $\varnothing$ and binary coproducts by $(A+B, \text{left}_{A,B}, \text{right}_{A,B}).$
Then, the category $\CCc$ has finite coproducts. Its initial object is $(\varnothing, \perp_{\varnothing, I})$ and binary coproducts are given by $(A, \diamond_A) +_a (B, \diamond_B) \coloneqq \left( A+B, [\diamond_A, \diamond_B] \right)$.
Moreover, we have: 
\[ + \circ (U \times U) = U \circ +_a  : \CCc \times \CCc \to \CC. \]
\end{proposition}
\begin{proof}
Straightforward verification.
\qed\end{proof}

\subsection{(Parameterised) initial algebras in $\CCc$}

In this subsection we will show how (parameterised) initial algebras from $\CC$ may be reflected into $\CCc$ by using methods from~\cite{lnl-fpc-lmcs,lnl-fpc-icfp}. Towards this end, we assume that $\CC$ has some additional structure, so that parameterised initial algebras may be formed within it.
\begin{assumption}
Throughout the remainder of the section, we assume that $\CC$ has an initial object $\varnothing$ and all $\omega$-colimits.
\end{assumption}
\begin{proposition}
\label{prop:ccc-colimits}
The category $\CCc$ has an initial object and all $\omega$-colimits. Moreover, the forgetful functor $U : \CCc \to \CC$ preserves and reflects $\omega$-colimits.
\end{proposition}
\begin{proof}
The initial object is $(\varnothing, \perp_{\varnothing, I})$, because $U$ reflects colimits (Proposition~\ref{prop:reflect}).

To show that $\CCc$ has all $\omega$-colimits, let $D: \omega \to \CCc$ be an arbitrary $\omega$-diagram of $\CCc$ with $D = \left( (D_0, \diamond_0) \xrightarrow{d_0} (D_1, \diamond_1) \xrightarrow{d_1} \cdots \right)$.
Let $\mu = (M, \mu_i : D_i \to M)$ be the colimiting cocone of $UD$ in $\CC$.
Using the discarding maps $\diamond_i$, we can now form a cocone $\diamond = (I, \diamond_i: D_i \to I)$ of $UD$ in $\CC$.
Let $\diamond_M : M \to I$ be the unique cocone morphism from $\mu$ to $\diamond$ induced by the colimit. It is now easy to see that we have a cocone $\tau = ( (M, \diamond_M) , \mu_i: (D_i, \diamond_i) \to (M, \diamond_M))$ of $D$ in $\CCc$. Clearly, $\mu = U \tau$ and since $U$ reflects colimits (Proposition~\ref{prop:reflect}), it follows that $\tau$
is the colimiting cocone of $D$ in $\CCc$. Therefore, $\CCc$ has $\omega$-colimits and by construction of the colimits, we see that $U$ preserves (and reflects) them.
\qed\end{proof}
Next, we show that the functor $U$ may be used to reflect $\omega$-cocontinuity of functors on $\CC$ to functors on $\CCc$.
\begin{theorem}
\label{thm:functor-reflect-cocontinuous}
Let $H: \CCc^n \to \CCc$ be a functor and $T: \CC^n \to \CC$ an $\omega$-cocontinuous functor, such that the diagram:
\cstikz{u-commute.tikz}
commutes. Then, $H$ is also $\omega$-cocontinuous.
\end{theorem}
\begin{proof}
Let $D: \omega \to \CCc^{n}$ be an arbitrary $\omega$-diagram in $\CCc^{n}$ and let $\mu$ be its colimiting cocone. Since $U$ preserves $\omega$-colimits (Proposition~\ref{prop:ccc-colimits}), it follows that $U^{\times n}\mu$ is a colimiting cocone of $U^{\times n}D$ in $\CC^n.$
By assumption $T$ is $\omega$-cocontinuous, so $TU^{\times n} \mu$ is a colimiting cocone of $TU^{\times n}D$ in $\CC$. By commutativity of the above diagram, it follows $UH \mu$ is a colimiting cocone of $UHD$ in $\CC$. But $U$ reflects colimits, so this means that $H \mu$ is a colimiting cocone of $HD$, as required.
\qed\end{proof}
Therefore, in the situation of the above theorem, both functors $H$ and $T$ have parameterised initial algebras by Theorem~\ref{thm:par-exists}.
This brings us to our next theorem.
\begin{theorem}
\label{thm:par-initial-algebras-same}
Let $H$ and $T$ be $\omega$-cocontinuous functors, such that the diagram 
\cstikz{u-commute-again.tikz}
commutes. Let $(T^\dagger, \phi)$ and $(H^\dagger, \psi)$ be their parameterised initial algebras.
Then: 
\begin{enumerate}
\item The following diagram:
\cstikz{dagger-commute.tikz}
commutes.
\item The following (2-categorical) diagram:
\cstikz{2-categorical.tikz}
commutes.
\end{enumerate}
\end{theorem}
\begin{proof}
The first statement follows by \cite[Corollary 4.21]{lnl-fpc-lmcs} and the second statement follows by \cite[Corollary 4.27]{lnl-fpc-lmcs}.
\qed\end{proof}
\begin{remark}
\label{rem:same-algebras}
The above theorem shows that the parameterised initial algebras of $H$ and $T$
respect the forgetful functor $U$ and are therefore constructed in the same way.
\end{remark}
\begin{remark}
If one is not interested in interpreting inductive data types defined by mutual
induction, then there is no need to form \emph{parameterised} initial algebras,
but merely initial algebras. In that case, the assumption that $\CC$ has all
$\omega$-colimits may be relaxed and one can assume that $\CC$ has colimits of
the initial sequences of the relevant functors. Then, most of the results
presented here can be simplified in a straightforward manner to handle this
case.
\end{remark}

\section{Categorical Model}\label{sec:model}

In this section we formulate our categorical model which we use to
interpret \lang{}.
\begin{notation}
We write $\cpo$ ($\cpobs$) for the category of (pointed) dcpo's and (strict) Scott-continuous maps between them.
\end{notation}
\begin{definition}
\label{def:model}
A categorical model of \lang{} is given by the following data:
\begin{enumerate}
\item A symmetric monoidal category $(\CC, \otimes, I, \alpha, \lambda, \rho, \sigma).$
\item An initial object $\varnothing \in \Ob(\CC)$ and binary coproducts $(A+B, \text{left}_{A,B}, \text{right}_{A,B}).$
\item The tensor product $\otimes$ distributes over $+$.
\item For each atomic type $\mathbf A \in \mathcal A$, an object $\mathbf A \in \Ob(\CC)$ together with a discarding map $\diamond_{\mathbf A} : \mathbf A \to I.$
\item The category $\CC$ has all $\omega$-colimits and $\otimes$ is an $\omega$-cocontinuous functor.
\item The category $\CC$ is $\cpobs$-enriched with least morphisms denoted $\perp_{A,B}$ and such that the symmetric monoidal structure and the coproduct structure are both $\cpo$-enriched.
\end{enumerate}
\end{definition}
This data suffices to interpret the language in the following way:
\begin{enumerate}
\item To interpret pair types.
\item To interpret sum types.
\item Used in the interpretation of \textbf{while} loops.
\item Necessary for the affine interpretation of the language.
\item To interpret inductive data types by forming parameterised initial algebras.
\item Used in the interpretation of \textbf{while} loops.
\end{enumerate}
\begin{assumption}
Henceforth, $\CC$ refers to an arbitrary, but fixed, model of \lang{}
and $\CCc \coloneqq \CC / I$ refers to the corresponding slice category
with the tensor unit $I$.
\end{assumption}
By using results from Section~\ref{sec:discarding}, we can now easily
establish some important properties of the category $\CCc.$ By
Proposition~\ref{prop:affine-monoidal}, it follows $\CCc$ has
a symmetric monoidal structure with tensor product $\otimes_a$
and by Proposition~\ref{prop:affine-coproducts}, it follows
$\CCc$ has finite coproducts with coproduct functor $+_a$.
We also know $\CCc$ has $\omega$-colimits by Proposition~\ref{prop:ccc-colimits}.
Finally, the next proposition is crucial for the construction of
discarding maps for inductive data types.
\begin{proposition}
\label{prop:tensor-coproduct-cocontinuous}
The functors $\otimes_a : \CCc \times \CCc$ and $+_a : \CCc \times \CCc \to \CCc$ are both
$\omega$-cocontinuous.
\end{proposition}
\begin{proof}
In the previous section we showed
$ \odot \circ (U \times U) = U \circ \odot_a  : \CCc \times \CCc \to \CC, $
for $\odot \in \{\otimes, +\}.$ Then, by Theorem~\ref{thm:functor-reflect-cocontinuous}, it follows $\odot$ is also $\omega$-cocontinuous.
\qed\end{proof}
Therefore, by Theorem~\ref{thm:par-exists}, we see that both categories $\CC$ and
$\CCc$ have sufficient structure to form parameterised initial algebras for all
functors composed out of tensors, coproducts and constants. The category $\CCc$
has the additional benefit that its parameterised initial algebras
also come equipped with discarding maps.

\subsection{Concrete models}

In this subsection we consider some concrete models of \lang{}. 
\begin{example}
The terminal category $\mathbf 1$ is a (completely degenerate) \lang{} model.
\end{example}
Next, we consider some non-degenerate models.
\begin{example}
The category $\cpobs$ is an \lang{} model.
\end{example}
However, in this model every object has a canonical comonoid structure, so it is not a truly representative
model for an affine type system like ours. In the next example we describe a
more representative model which has been studied in the context of circuit
description languages and quantum programming.
\begin{example}
Let $\mathbf M$ be a small $\dcpobs$-symmetric monoidal category and 
let $\widehat{\mathbf M} = [\mathbf M^\op, \dcpobs]$ be the indicated $\dcpobs$-functor category.
Then $\widehat{\mathbf M}$ is an \lang{} model when equipped with the Day convolution monoidal structure~\cite{day-convolution}.
By making suitable choices for $\mathbf M$, the category $\widehat{\mathbf M}$
becomes a model of Proto-Quipper-M~\cite{pqm} and ECLNL~\cite{eclnl}, which are
programming languages for string diagrams that have also been studied in the
context of quantum programming.
\end{example}
Next, we discuss how fragments of the language may be interpreted in categories of W*-algebras~\cite{takesaki}, which are used to study quantum computing.
\begin{example}
Let $\WNMIU$ be the category of W*-algebras and normal unital $*$-homomorphisms between them. Let $\VV \coloneqq (\WNMIU)^\op$ be its opposite category.
Then $\VV$ is an \lang{} model without recursion~\cite{quantum-collections}, i.e., one can interpret all \lang{} constructs except for while loops within $\VV$, because $\VV$ is not $\dcpobs$-enriched.
\end{example}
\begin{example}
Let $\Wstar$ be the category of W*-algebras and normal completely-positive subunital maps between them. Let $\DD \coloneqq (\Wstar)^\op$ be its opposite category.
Then $\DD$ is an \lang{} model which supports simple non-nested type induction, i.e., one can interpret all \lang{} constructs within $\DD$ using the methods described in this paper, provided that inductive data types contain at most one free type variable.
\end{example}
\begin{remark}
In fact, it is possible to interpret all of QPL (and therefore also \lang{}
which is a fragment of QPL) by using an adjunction between $\VV$ and $\DD$, as
was shown in~\cite{qpl-fossacs}. However, this requires considering the
specifics of this particular model, which has not been axiomatised yet, and
separating the interpretation of types and values (in $\VV$) from the
interpretation of terms (in $\DD$).
\end{remark}

\section{Denotational Semantics of \lang{}}\label{sec:semantics}
In this section we present the denotational semantics of \lang{}. First, we
show how types are interpreted  in \secref{sub:types}. Since our type system is
affine, we construct discarding maps for all types in
\secref{sub:affine}. Folding and unfolding of inductive types are shown to be
discardable operations in \secref{sub:folding}. The interpretations of terms
and configurations are defined in \secref{sub:terms} and
\secref{sub:configurations}. Finally, we prove soundness and adequacy in
\secref{sub:soundness}.

\subsection{Interpretation of types}\label{sub:types}
The (standard) interpretation of a type $\Theta \vdash A$ is a functor 
$\lrb{\Theta \vdash A} : \CC^{|\Theta|} \to \CC$, defined in Figure~\ref{fig:type-interpretations} (left), where
$K_X$ indicates the constant $X$-functor. 
\begin{figure}[t]
  \begin{minipage}{0.48\textwidth}
\small
   \begin{align*}
    \lrb{\Theta \vdash A}           &: \CC^{|\Theta|} \to \CC \\
    \lrb{\Theta \vdash \Theta_i}    &= \Pi_i \\
    \lrb{\Theta \vdash I}           &= K_{I} \\
    \lrb{\Theta \vdash \mathbf A}   &= K_{\mathbf A} \\
    \lrb{\Theta \vdash A + B}       &= + \circ \langle \lrb{\Theta \vdash A} , \lrb{\Theta \vdash B} \rangle \\
    \lrb{\Theta \vdash A \otimes B} &= \otimes \circ \langle \lrb{\Theta \vdash A} , \lrb{\Theta \vdash B} \rangle \\
    \lrb{\Theta \vdash \mu X. A}    &= \lrb{\Theta, X \vdash A}^\dagger
    \end{align*}
\end{minipage}
 \vline 
 \quad
  \begin{minipage}{0.48\textwidth}
  \small
    \begin{align*}
    \elrb{\Theta \vdash A}           &: \CCc^{|\Theta|} \to \CCc \\
    \elrb{\Theta \vdash \Theta_i}    &= \Pi_i \\
    \elrb{\Theta \vdash I}           &= K_{(I, \id_I)} \\
    \elrb{\Theta \vdash \mathbf A}   &= K_{(\mathbf A, \diamond_{\mathbf A})} \\
    \elrb{\Theta \vdash A + B}       &= +_a \circ \langle \elrb{\Theta \vdash A} , \elrb{\Theta \vdash B} \rangle \\
    \elrb{\Theta \vdash A \otimes B} &= \otimes_a \circ \langle \elrb{\Theta \vdash A} , \elrb{\Theta \vdash B} \rangle \\
    \elrb{\Theta \vdash \mu X. A}    &= \elrb{\Theta, X \vdash A}^\dagger
    \end{align*}
  \end{minipage}
\caption{Standard (left) and affine (right) interpretations of types. }
\label{fig:type-interpretations}
\end{figure}
We begin by showing that this assignment is well-defined, i.e., we have to show that the required parameterised initial algebras exist.
\begin{proposition}
\label{prop:well-defined}
$\lrb{\Theta \vdash A}$ is a well-defined $\omega$-cocontinuous functor.
\end{proposition}
\begin{proof}
Projection functors and constant functors are obviously $\omega$-cocontinuous. The coproduct functor is $\omega$-cocontinuous, because it is given by a colimiting construction. The tensor product $\otimes$ is $\omega$-cocontinuous by assumption. Also, $\omega$-cocontinuous functors are closed under composition and pairing~\cite{ls81}.
By Theorem~\ref{thm:par-exists}, $\lrb{\Theta, X \vdash A}^\dagger$ is well-defined and also an $\omega$-cocontinuous functor.
\qed\end{proof}
The semantics of terms is defined on closed types, so for brevity we introduce the following notation.
\begin{notation}
For any closed type $\cdot \vdash A$, let
$\lrb{A} \coloneqq \lrb{\cdot \vdash A}(*) \in \Ob(\CC) , $ 
where $*$ indicates the only object of the terminal category $\textbf{1}$.
\end{notation}

\subsection{Affine Structure of Types}\label{sub:affine}

In this subsection we describe the affine structure of our types by
constructing an appropriate discarding map for every type. This is achieved by using the
results we established in~\secref{sec:discarding} and by providing an
\emph{affine interpretation of types} as functors on the slice category $\CCc =
\CC / I.$ The affine interpretation is related to the standard one via the
forgetful functor which results in the construction of the required discarding
maps.

The affine interpretation of a type $\Theta \vdash A$ is a functor 
$\elrb{\Theta \vdash A} : \CCc^{|\Theta|} \to \CCc$, defined in Figure~\ref{fig:type-interpretations} (right). 
\begin{proposition}
$\elrb{\Theta \vdash A}$ is a well-defined $\omega$-cocontinuous functor.
\end{proposition}
\begin{proof}
The tensor product $\otimes_a$ and coproduct functors $+_a$ are $\omega$-cocontinuous by Proposition~\ref{prop:tensor-coproduct-cocontinuous}. Using the same arguments as in Proposition~\ref{prop:well-defined}, we finish the proof.
\qed\end{proof}
\begin{notation}
For any closed type $\cdot \vdash A$, let
$ \elrb{A} \coloneqq \elrb{\cdot \vdash A}(*) \in \Ob(\CCc) . $
\end{notation}
We proceed by describing the relationship between the standard and affine interpretation of types.

\begin{theorem}\label{thm:discard}
For any type $\Theta \vdash A$, the following diagram
\cstikz{types-relation.tikz}
commutes. Therefore, for any closed type $\cdot \vdash A$, we have $\lrb A = U \elrb A.$
\end{theorem}
\begin{proof}
By induction on $\Theta \vdash A$ using the established results from~\secref{sec:discarding}.
\qed\end{proof}

This theorem shows that for any closed type $A$, we have $\elrb A = (\lrb A,
\diamond_{\lrb A}),$ where the discarding map $\diamond_{\lrb A} : \lrb A \to
I$ is constructed by the affine type interpretation in
Figure~\ref{fig:type-interpretations} (right).  We will later see
(Theorem~\ref{thm:affine-values}) that the interpretations of our values are
discardable morphisms with respect to this choice of discarding maps.

\subsection{Folding and Unfolding of Inductive Datatypes}\label{sub:folding}

The purpose of this subsection is to define \emph{folding} and \emph{unfolding}
of inductive data types (which we need to define the term semantics) and also to
demonstrate that folding/unfolding is a discaradble isomorphism with respect to
the affine structure of our types.
\begin{lemma}[Type Substitution]
\label{lem:type-sub}
Let $\Theta, X \vdash A$ and $\Theta \vdash B$ be types. Then:
\begin{enumerate}
  \item $ \lrb{\Theta \vdash A[B/X]} = \lrb{\Theta, X \vdash A} \circ \langle \Id, \lrb{\Theta \vdash B}\rangle . $
  \item $ \elrb{\Theta \vdash A[B/X]} = \elrb{\Theta, X \vdash A} \circ \langle \Id, \elrb{\Theta \vdash B}\rangle . $
\end{enumerate}
\end{lemma}
\begin{proof}
Straightforward induction, essentially the same as \cite[Lemma 6.5]{lnl-fpc-lmcs}.
\qed\end{proof}
\begin{definition}\label{def:fold-unfold}
For any closed type $\cdot \vdash \mu X. A,$ we define two isomorphisms:
\begin{align*}
\sfold_{\mu X.A} &: \lrb{A[\mu X. A/ X]}  = \lrb{X \vdash A} \lrb{\mu X. A}   \cong \lrb{\mu X. A}  : \sunfold_{\mu X.A} \\
\efold_{\mu X.A} &: \elrb{A[\mu X. A/ X]}  = \elrb{X \vdash A} \elrb{\mu X. A}   \cong \elrb{\mu X. A}  : \eunfold_{\mu X.A}
\end{align*}
\end{definition}
Since type substitution holds up to equality, it follows that folding/unfolding
of inductive data types is determined entirely by the initial algebra structure
of the corresponding endofunctors. Finally, we show that folding/unfolding of
an inductive data type is the same isomorphism for both the standard and affine
type interpretations.
\begin{theorem}\label{thm:fold-unfold-relationship}
Given a closed type $\cdot \vdash \mu X. A$, then the following diagram
\cstikz{fold-unfold.tikz}
commutes.
\end{theorem}
\begin{proof}
This follows immediately by Theorem~\ref{thm:par-initial-algebras-same} (2).
\qed\end{proof}
Therefore folding/unfolding of types is a discardable isomorphism.
%
\subsection{Interpretation of terms}\label{sub:terms}
In this subsection we explain how to interpret the terms of $\lang{}.$

A variable context $\Gamma = x_1:A_1, \ldots, x_n:A_n$ is interpreted as the object $\lrb{\Gamma} := \lrb{A_1} \otimes \cdots \otimes \lrb{A_n} \in \text{Ob}(\CC).$
A term judgement $ \vdash \langle \Gamma \rangle\ M\ \langle \Sigma \rangle$ is interpreted as a morphism $ \lrb{ \vdash \langle \Gamma \rangle\ M\ \langle \Sigma \rangle} : \lrb \Gamma \to \lrb \Sigma $
of $\CC$ which is defined in Figure~\ref{fig:semantics}.
For brevity, we will simply write $\lrb{M} \coloneqq \lrb{\Pi \vdash \langle \Gamma \rangle\ M\ \langle \Sigma \rangle}$ whenever the contexts are clear.
\begin{figure}[t]
$
\small{
\begin{array}{l}
\lrb{ \vdash \langle \Gamma \rangle\  \newunit\ u\ \langle \Gamma, u:I \rangle} \coloneqq  \left(
      \lrb{\Gamma}
        \xrightarrow{\cong}
      \lrb{\Gamma} \otimes I
      \right)\\
\lrb{ \vdash \langle \Gamma, x : A \rangle\  \textbf{discard}\ x\ \langle \Gamma \rangle} \coloneqq  \left(
      \lrb{\Gamma} \otimes \lrb A
        \xrightarrow{\id \otimes \diamond}
      \lrb{\Gamma} \otimes I
        \xrightarrow{\cong}
      \lrb{\Gamma}
      \right)\\
\lrb{ \vdash \langle \Gamma \rangle\ M;N\ \langle \Sigma \rangle} \coloneqq  \left(
      \lrb{\Gamma}
        \xrightarrow{\lrb{M}}
      \lrb{\Gamma'}
        \xrightarrow{\lrb N}
      \lrb{\Sigma}
      \right)\\
\lrb{ \vdash \langle \Gamma \rangle\ \textbf{skip}\ \langle \Gamma \rangle} \coloneqq  \left(
      \lrb{\Gamma}
        \xrightarrow{\id}
      \lrb{\Gamma}
      \right)\\
\lrb{ \vdash \langle \Gamma, b: \textbf{bit} \rangle\ \textbf{while}\ b\ \textbf{do}\ {M}\ \langle \Gamma, b: \textbf{bit} \rangle} \coloneqq  \left(
      \lrb{\Gamma} \otimes \textbf{bit} 
        \xrightarrow{\mathrm{lfp}(W_{\lrb M})}
      \lrb{\Gamma} \otimes \textbf{bit} \right)\\
\lrb{ \vdash \langle \Gamma, x:A \rangle\ y = \lleft_{A,B}\ x\ \langle \Gamma, y: A+B \rangle} \coloneqq  \left(
      \lrb{\Gamma} \otimes \lrb A
        \xrightarrow{ \id \otimes \mathrm{left}_{A,B}}
      \lrb \Gamma \otimes (\lrb A + \lrb B) 
\right)\\
\lrb{ \vdash \langle \Gamma, x:B \rangle\ y = \rright_{A,B}\ x\ \langle \Gamma, y: A+B \rangle} \coloneqq  \left(
      \lrb{\Gamma} \otimes \lrb B
        \xrightarrow{ \id \otimes \mathrm{right}_{A,B}}
      \lrb \Gamma \otimes (\lrb A + \lrb B)
\right)\\
\llbracket  \vdash \langle \Gamma, y: A+B \rangle\  \ccase\ y\ \textbf{of}\ \{\lleft\ x_1 \to M_1\ |\ \rright\ x_2 \to M_2\}\ \langle \Sigma \rangle \rrbracket \coloneqq \\
\qquad \qquad  \left( \lrb \Gamma \otimes (\lrb{A} + \lrb B) \xrightarrow{d} (\lrb \Gamma \otimes \lrb A) + (\lrb \Gamma \otimes \lrb B) \xrightarrow{\left[\lrb {M_1}, \lrb {M_2} \right]} \lrb \Sigma \right)\\
\lrb{ \vdash \langle \Gamma, x_1: A, x_2: B \rangle\ x=(x_1, x_2) \ \langle \Gamma, x: A \otimes B \rangle} \coloneqq  \left(
      \lrb{\Gamma} \otimes \lrb A \otimes \lrb B
        \xrightarrow{\id}
      \lrb{\Gamma} \otimes \lrb A \otimes \lrb B \right)\\
\lrb{ \vdash \langle \Gamma, x: A \otimes B \rangle\ (x_1, x_2) = x \ \langle \Gamma, x_1: A, x_2: B \rangle} \coloneqq  \left(
      \lrb{\Gamma} \otimes \lrb A \otimes \lrb B
        \xrightarrow{\id}
      \lrb{\Gamma} \otimes \lrb A \otimes \lrb B \right)\\
\lrb{ \vdash \langle \Gamma, x: A[\mu X. A / X] \rangle\ y = \textbf{fold}\ x\ \langle \Gamma, y: \mu X.A \rangle} \coloneqq  \left(
      \lrb{\Gamma} \otimes \lrb{A[\mu X. A / X]}
        \xrightarrow{\id \otimes \sfold}
      \lrb \Gamma \otimes \lrb{\mu X.A} \right)\\
\lrb{ \vdash \langle \Gamma, x: \mu X. A \rangle\ y = \textbf{unfold}\ x\ \langle \Gamma, y: A[\mu X. A / X] \rangle} \coloneqq  \left(
      \lrb{\Gamma} \otimes \lrb{\mu X. A}
        \xrightarrow{\id \otimes \sunfold}
      \lrb \Gamma \otimes \lrb{A[\mu X. A / X]} \right)\\
\end{array}
}
$
\caption{Interpretation of \lang{} terms.}
\label{fig:semantics}
\end{figure}

Next, we clarify some of the notation used in Figure~\ref{fig:semantics}.
The map $\diamond_{\lrb A}$ is defined in~\secref{sub:affine}, as already explained.
In order to interpret \textbf{while} loops, we use a Scott-continuous endofunction $W_f$, which is defined as follows.
For a morphism $f:{A \otimes \bit \to A \otimes \bit}$, we set:
\begin{align*}
W_{f} &: \CC \left(A \otimes \bit, A \otimes \bit \right) \to \CC(A \otimes \bit, A \otimes \bit)\\
W_{f}(g) &= \left[ \id \otimes \mathrm{left}_{I,I},\ g \circ f \circ (\id \otimes \mathrm{right}_{I,I}) \right] \circ d_{A,I,I} ,
\end{align*}
where $d_{A,I,I} : A \otimes (I+I) \to (A \otimes I) + (A \otimes I)$ is the isomorphism which is induced by the distributivity of $\otimes$ over $+$ (see Definition~\ref{def:model}).
Finally, for a pointed dcpo $D$ and Scott-continuous endofunction $h : D \to D$, the \emph{least fixpoint} of $h$ is given by
$\mathrm{lfp}(h) \coloneqq \bigvee_{i=0}^\infty h^i(\perp),$ where $\perp$ is the least element of $D$.

\subsection{Interpretation of configurations}\label{sub:configurations}
Before we explain how to interpret configurations, we have to show how to
interpret values.
\paragraph{Interpretation of values.}
The interpretation of a value $ \vdash v : A$ is a morphism
$\lrb{ \vdash v: A} : I \xrightarrow{} \lrb A,$ and we shall simply write
$\lrb v$ if its type is clear from context. The interpretation is defined in Figure~\ref{fig:value-interpretation}.
\begin{figure}[t]
\begin{align*}
  \lrb{\cdot \vdash * : I}                    &\coloneqq \id_I\\
  \lrb{Q \vdash \lleft_{A,B} v: A+B}          &\coloneqq \mathrm{left} \circ \lrb v\\
  \lrb{Q \vdash \rright_{A,B} v: A+B}         &\coloneqq \mathrm{right} \circ \lrb v\\
  \lrb{Q_1, Q_2 \vdash (v, w) : A \otimes B}  &\coloneqq ( \lrb v \otimes \lrb w ) \circ \lambda_I^{-1} \\
  \lrb{Q \vdash \fold_{\mu X.A} v : \mu X. A} &\coloneqq \mathrm{fold} \circ \lrb v
\end{align*}
\caption{Interpretation of \lang{} values.}
\label{fig:value-interpretation}
\end{figure}
In order to prove soundness of our \emph{affine} type system, we have to show every value is discardable.
\begin{theorem}\label{thm:affine-values}
For every value $\vdash v : A$, we have: $\diamond_{\lrb A} \circ \lrb{\vdash v : A} = \id_I.$
\end{theorem}
\begin{proof}
By construction, $\diamond_{\lrb A}$ enjoys all of the properties established in
\secref{sec:discarding}. The proof proceeds by induction on the derivation of
$\ \vdash v : A.$ The base case is trivial. Discardable morphisms are closed under composition, because $\CCc$ is a category. Moreover, discardable maps are closed under tensor
products (Proposition~\ref{prop:affine-monoidal}).  Using
the induction hypothesis, it suffices to show that the coproduct injections and
folding are discardable maps. But this follows by
Proposition~\ref{prop:affine-coproducts} and
Theorem~\ref{thm:fold-unfold-relationship}, respectively.
\qed\end{proof}
Given a variable context $\Gamma = x_1 : A_1, \ldots, x_n : A_n $, then a value context $ \Gamma \vdash V$ is interpreted by the morphism:
\[ \lrb{ \Gamma \vdash V} = \left( I \xrightarrow{\cong} I^{\otimes n} \xrightarrow{\lrb{v_1} \otimes \cdots \otimes \lrb{v_n}} \lrb \Gamma \right) , \]
where $V = \{ x_1 = v_1, \ldots , x_n = v_n \} $
and we abbreviate this by writing $\lrb{V}$. Note that $\lrb V$ is also discardable due to Theorem~\ref{thm:affine-values}.
\paragraph{Interpretation of configurations.}
A configuration $\Gamma; \Sigma \vdash (M\ |\ V) $ is interpreted as the morphism
\[  \lrb{\Gamma; \Sigma \vdash (M\ |\ V)} =  \left( I \xrightarrow{\lrb{ \Gamma \vdash V}} \lrb \Gamma \xrightarrow{\lrb{ \vdash \langle \Gamma \rangle\ M\ \langle \Sigma \rangle}} \lrb \Sigma \right) . \]
We write $\lrb{(M\ |\ V)}$ for this morphism whenever the contexts are clear.

\subsection{Soundness and Computational Adequacy}\label{sub:soundness}
Soundness is the statement that the denotational semantics is invariant under program execution.
\begin{theorem}[Soundness]\label{thm:soundness}
If $\mathcal C \leadsto \mathcal D$, then $ \lrb{\mathcal C} = \lrb{\mathcal D}. $
\end{theorem}
\begin{proof}
Straightforward induction.
\qed\end{proof}
We conclude our technical contributions by proving a computational adequacy result.
Towards this end, we have to assume that our categorical model is not degenerate.
\begin{definition}
A \emph{computationally adequate} \lang{} model is an \lang{} model, where $\id_I \neq \perp.$ 
\end{definition}
A program configuration $\mathcal C$ is said to \emph{terminate}, denoted $\mathcal C \Downarrow$,
if there exists a terminal configuration $\mathcal T,$ such that $\mathcal C \leadsto_* \mathcal T$,
where $\leadsto_*$ is the reflexive and transitive closure of $\leadsto.$
\begin{theorem}[Adequacy]
Let $\vdash \langle \cdot \rangle\ M\ \langle \Sigma \rangle$ be a closed term. Then:
\[ \lrb{M} \neq \perp \text{ iff } (M\ |\ \cdot) \Downarrow. \]
\end{theorem}
\begin{proof}
By simplifying the adequacy proof strategy of QPL~\cite{qpl-fossacs} in the obvious way.
\qed\end{proof}

\section{Future Work}
\label{sec:future}

As part of future work it will be interesting to see whether these methods can
be adapted to also work with coinductive data types and/or with recursive
data types where function types $(\multimap)$ become admissible constructs
within the type recursion schemes. This is certainly a more challenging problem
which would probably require us to assume additional structure within the
model, such as a limit-colimit coincidence~\cite{smyth-Plotkin}, so that we may
deal with the contravariance induced by function types. It is also likely that
we would have to modify the slice construction to accommodate the addition of
limits.

\section{Conclusion and Related Work}
\label{sec:conclude}

Since the introduction of Linear Logic~\cite{linear-logic}, there has been a
massive amount of research into finding suitable models for (fragments) of
Linear Logic (see~\cite{mellies-logic} for an excellent overview).  However,
there has been less research into models of affine logics and affine type
systems. The principle difference between linear and affine logic is that
weakening is restricted in the former, but allowed in the latter, so affine
models have to contain additional discarding maps. Most models of affine type
systems (that I am aware of) use some specific properties of the model, such as
finding a suitable (sub)category where the tensor unit $I$ is a terminal
object, in order to construct the required discarding maps (e.g.
\cite{affine-games,quantum-games,qpl-fossacs,qlc-affine}). This means, the solution is provided directly by the \emph{model}.
In this paper, we have taken a different approach, because we present a
\emph{semantic} solution to this problem. The only assumption that we have made
for our model is that the interpretations of the atomic types are equipped with
suitable discarding maps and we then show how to construct all other required
discarding maps by considering an additional and non-standard interpretation
of types within a slice category. Overall, the "model" solution is certainly
simpler and more concise compared to the "semantic" solution presented here. On
the other hand, our solution in this paper is very general and can in principle
be applied to models where the required discarding maps are unknown a priori.

\paragraph{Acknowledgements.} I thank Romain P\'echoux, Simon Perdrix and
Mathys Rennela for discussions about the methods in this paper. I also
gratefully acknowledge financial support from the French projects
ANR-17-CE25-0009 SoftQPro and PIA-GDN/Quantex.

\bibliography{refs}

\begin{thebibliography}{10}
\providecommand{\url}[1]{\texttt{#1}}
\providecommand{\urlprefix}{URL }
\providecommand{\doi}[1]{https://doi.org/#1}

\bibitem{quantum-games}
Clairambault, P., de~Visme, M., Winskel, G.: Game semantics for quantum
  programming. {PACMPL}  \textbf{3}({POPL}),  32:1--32:29 (2019).
  \doi{10.1145/3290345}

\bibitem{day-convolution}
Day, B.: On closed categories of functors ii. In: Category Seminar: Proceedings
  Sydney Category Theory Seminar 1972/1973. pp. 20--54 (1974).
  \doi{10.1007/BFb0063098}

\bibitem{fiore-thesis}
Fiore, M.P.: Axiomatic domain theory in categories of partial maps. Ph.D.
  thesis, University of Edinburgh, {UK} (1994)

\bibitem{linear-logic}
Girard, J.: Linear logic. Theor. Comput. Sci.  \textbf{50},  1--102 (1987).
  \doi{10.1016/0304-3975(87)90045-4}

\bibitem{quantum-collections}
Kornell, A.: Quantum collections. International Journal of Mathematics
  \textbf{28}(12),  1750085 (2017). \doi{10.1142/S0129167X17500859}

\bibitem{affine-games}
Laird, J.: A game semantics of linearly used continuations. In: Gordon, A.D.
  (ed.) Foundations of Software Science and Computational Structures, 6th
  International Conference, {FOSSACS} 2003 Held as Part of the Joint European
  Conference on Theory and Practice of Software, {ETAPS} 2003, Warsaw, Poland,
  April 7-11, 2003, Proceedings. Lecture Notes in Computer Science, vol.~2620,
  pp. 313--327. Springer (2003). \doi{10.1007/3-540-36576-1\_20}

\bibitem{ls81}
Lehmann, D.J., Smyth, M.B.: Algebraic specification of data types: A synthetic
  approach. Mathematical Systems Theory  (1981)

\bibitem{lnl-fpc-icfp}
Lindenhovius, B., Mislove, M., Zamdzhiev, V.: Mixed linear and non-linear
  recursive types. Proceedings of the ACM on Programming Languages
  \textbf{3}(ICFP),  111:1--111:29 (Jul 2019). \doi{10.1145/3341715}

\bibitem{eclnl}
Lindenhovius, B., Mislove, M.W., Zamdzhiev, V.: Enriching a linear/non-linear
  lambda calculus: {A} programming language for string diagrams. In: Dawar, A.,
  Gr{\"{a}}del, E. (eds.) Proceedings of the 33rd Annual {ACM/IEEE} Symposium
  on Logic in Computer Science, {LICS} 2018, Oxford, UK, July 09-12, 2018. pp.
  659--668. {ACM} (2018). \doi{10.1145/3209108.3209196}

\bibitem{lnl-fpc-lmcs}
Lindenhovius, B., Mislove, M.W., Zamdzhiev, V.: Lnl-fpc: The linear/non-linear
  fixpoint calculus (2020), accepted subject to minor revisions for the journal
  LMCS (Logical Methods in Computer Science). Available at
  \url{http://arxiv.org/abs/1906.09503}

\bibitem{mellies-logic}
Mellies, P.A.: Categorical semantics of linear logic. Panoramas et syntheses
  (2009)

\bibitem{qpl-fossacs}
P{\'{e}}choux, R., Perdrix, S., Rennela, M., Zamdzhiev, V.: Quantum programming
  with inductive datatypes: Causality and affine type theory. In: FoSSaCS
  (Foundations of Software Science and Computation Structures), to appear
  (2020), available at \url{http://arxiv.org/abs/1910.09633}

\bibitem{pqm}
Rios, F., Selinger, P.: A categorical model for a quantum circuit description
  language. In: {QPL} 2017 (2017). \doi{10.4204/EPTCS.266.11}

\bibitem{qpl}
Selinger, P.: Towards a quantum programming language. Mathematical Structures
  in Computer Science  \textbf{14}(4),  527--586 (2004).
  \doi{10.1017/S0960129504004256}

\bibitem{qlc-affine}
Selinger, P., Valiron, B.: Quantum Lambda Calculus (2009).
  \doi{10.1017/CBO9781139193313.005}

\bibitem{smyth-Plotkin}
Smyth, M.B., Plotkin, G.D.: The category-theoretic solution of recursive domain
  equations. {SIAM} J. Comput.  \textbf{11}(4),  761--783 (1982).
  \doi{10.1137/0211062}, \url{https://doi.org/10.1137/0211062}

\bibitem{takesaki}
Takesaki, M.: {Theory of operator algebras. Vol. I, II and III}.
  Springer-Verlag, Berlin (2002)

\bibitem{no-cloning}
Wootters, W.K., Zurek, W.H.: A single quantum cannot be cloned. Nature
  \textbf{299}(5886),  802--803 (1982)

\end{thebibliography}

\appendix
\newpage

\appendix

\section{Omitted Proofs from Section~\ref{sec:discarding}}

\subsection{Proof of Proposition~\ref{prop:reflect}}\label{proof:reflect}

\begin{proof}
Let $D : \JJ \to \CCc$ be a diagram and let $\mu = ( (M, \diamond_M), \mu_i : (D_i, \diamond_i) \to (M, \diamond_M) )$ be a cocone of $D$ in $\CCc$, such that $U \mu = (M, \mu_i : D_i \to M)$ is a colimiting cocone of $UD$ in $\CC$.
Let $\chi = ( (X, \diamond_X), \chi_i : (D_i, \diamond_i) \to (X, \diamond_X) )$ be another cocone of $D$ in $\CCc.$ We will show there exists a unique cocone morphism $f : \mu \to \chi$ in $\CCc$.

By the universal property of the colimit in $\CC$, we see there exists a unique cocone morphism $f : U\mu \to U \chi$ of $UD$ in $\CC$, i.e. $f: M \to X$ is the unique morphism of $\CC$, such that
\begin{equation}\label{eq:unique-reflect}
f \circ \mu_i = \chi_i \qquad \forall i \in \mathrm{Ob}(\JJ) 
\end{equation}
Since $D: \JJ \to \CCc$ is a diagram, then for any $f : i \to j$ in $\JJ$, we have $D_f : (D_i, \diamond_i) \to (D_j, \diamond_j)$ and thus $\diamond_j \circ D_f = \diamond_i$.
This means we have a cocone $\diamond = (I, \diamond_i : D_i \to I)$ of $UD$ in $\CC$.
By the universal property of the colimit, there exists a unique cocone morphism $g : U \mu \to \diamond,$ i.e. $g: M \to I$ is the unique morphism of $\CC$, such that
\begin{equation}\label{eq:unique-diamond}
g \circ \mu_i = \diamond_i \qquad \forall i \in \mathrm{Ob}(\JJ) 
\end{equation}
But $\mu_i : (D_i, \diamond_i) \to (M, \diamond_M)$ in $\CCc,$ therefore $\diamond_M \circ \mu_i = \diamond_i$ and thus $g = \diamond_M$.
However:
\begin{align*}
   (\diamond_X \circ f) \circ \mu_i
&= \diamond_X \circ (f \circ \mu_i) & \\ 
&= \diamond_X \circ \chi_i &\eqref{eq:unique-reflect} \\ 
&= \diamond_i &(\chi_i : (D_i, \diamond_i) \to (X, \diamond_X) \text{ in } \CCc)
\end{align*}
And therefore by~\eqref{eq:unique-diamond} it follows $ \diamond_X \circ f = g = \diamond_M.$
But this now shows that $f: (M, \diamond_M) \to (X, \diamond_X)$ is a morphism
of $\CCc$. Obviously, $f: \mu \to \chi$ is also a cocone morphism of $D$, because
composition in $\CCc$ coincides with composition in $\CC$.

Finally, to show $f$ is unique, assume that $h: \mu \to \chi$ is a cocone morphism of $D$. But then $h: U\mu \to U\chi$ is a cocone morphism
of $UD$ in $\CC$ and therefore $f = h$.
\end{proof}

\end{document}